\documentclass[11pt]{article}
\usepackage{amsmath,amsthm,amsfonts,amscd,eucal,latexsym,amssymb,mathrsfs}
\usepackage{amsxtra, calc, bbm}
\usepackage{epsfig}  
\usepackage[latin1]{inputenc}
\usepackage{verbatim} 
\def\cA{{\ca A}}

\def\cC{{\ca C}}
\def\cD{{\ca D}}

\def\cH{{\ca H}}

\def\cO{{\ca O}}

\def\cR{{\ca R}}

\def\mA{{\mathcal A}}

\def\bR{{\mathbb R}}

\def\bS{{\mathbb S}}

\def\beq{\begin{eqnarray}}
\def\eeq{\end{eqnarray}}
\def\pa{\partial}
\def\at{\left(}               
\def\aq{\left[}               

\def\ct{\right)}              
\def\cq{\right]}              
\newcommand{\ca}[1]{{\cal #1}}         

\def\be{\beta}
\def\ga{\gamma}
\def\de{\delta}

\def\la{\lambda}

\def\si{\sigma}
\def\om{\omega}

\def\De{\Delta}

\def\Si{\Sigma}
\def\Om{\Omega}

\def\m{\mu}
\def\n{\nu}
\newcommand{\ahat}{\check}

\newcommand{\sE}{\mathscr{E}}
\newcommand{\sZ}{\mathscr{Z}}

\newtheorem{theorem}{Theorem}[section]
\newtheorem{proposition}{Proposition}[section]

\numberwithin{equation}{section}

\newcommand{\se}[1]{\section{#1}}

\def\vsp{\vspace{0.2cm}}
\def\vspp{\vspace{0.1cm}}

\def\sse #1 {\vsp\ifhmode{\par}\fi\refstepcounter{subsection}
  \noindent {\bf\thesubsection}. {\em #1}.\quad
  \addcontentsline{toc}{subsection}{\protect\numberline{\thesubsection} #1}%
  }

\def\ssb #1 {\vsp\ifhmode{\par}\fi\refstepcounter{subsection}
  \noindent {\bf\thesubsection.} {\bf #1.}\quad
  \addcontentsline{toc}{subsection}{\protect\numberline{\thesubsection} #1}%
  }

\def\ssa #1 {\ifhmode{\par}\fi\refstepcounter{subsection}
  \noindent {\bf\thesubsection.} {\bf #1.}\quad
  \addcontentsline{toc}{subsection}{\protect\numberline{\thesubsection} #1}%
  }

\def\remark #1 {\vsp\vspp\ifhmode{\par}\fi\noindent\noindent {\bf Remark.} {#1}\vsp\vspp\par}
\def\remarks #1 {\vsp\vspp\ifhmode{\par}\fi\noindent\noindent {\bf Remarks.} {#1}\vsp\vspp\par}

\begin{document}


\par
\bigskip
\LARGE
\noindent
{\bf On Quantum Spacetime and the horizon problem }
\bigskip
\par
\rm
\normalsize


\large
\noindent {\bf Sergio Doplicher$^{1}$}, {\bf Gerardo Morsella$^{2}$}
and {\bf Nicola Pinamonti$^{3}$}\\
\par
\small
\noindent $^1$
Dipartimento di Matematica, Universit\`a di Roma ``La Sapienza'', Piazzale Aldo Moro, 5,
I-00185 Roma, Italy, email dopliche@mat.uniroma1.it.\smallskip

\noindent $^2$
Dipartimento di Matematica, Universit\`a di Roma ``Tor Vergata'', Via della Ricerca Scientifica,
I-00133 Roma, Italy, email morsella@mat.uniroma2.it.\smallskip

\noindent $^3$
Dipartimento di Matematica, Universit\`a di Genova, Via Dodecaneso, 35, 
I-16146 Genova, Italy, email pinamont@dima.unige.it.\smallskip

 \normalsize
\par
\medskip

\rm\normalsize
\noindent {\small \today}

\rm\normalsize


\par
\bigskip

\noindent
\small
{\bf Abstract}.

In the special case of a spherically symmetric solution of Einstein equations coupled to a scalar massless field, we examine the consequences on the exact solution imposed by a semiclassical treatment of gravitational interaction when the scalar field is quantized. In agreement with~\cite{DFR}, imposing the {\it  principle of gravitational stability against localization of events}, we find that  the region where an event is localized, or where initial conditions can be assigned, has a minimal extension, of the order of the Planck length. This conclusion, though limited to the case of spherical symmetry, is more general than that of~\cite{DFR} since it does not require the use of the notion of energy through the Heisenberg Principle, nor of any approximation as the linearized Einstein equations. 

We shall then describe the influence of this minimal length scale in a cosmological model, 
namely a simple universe filled with radiation, which is effectively described by a conformally coupled scalar field in a conformal KMS state. 
Solving the backreaction, a power law inflation scenario appears close to the initial singularity.  Furthermore, the initial singularity becomes light like and thus  the standard horizon problem is avoided in this simple model. 
This indication goes in the same direction as those drawn at a heuristic level from a full use of the  {\it  principle of gravitational stability against localization of events}, which point to a {\it background dependence of the effective
Planck length}, through which a-causal effects may be transmitted.

\normalsize

\tableofcontents

\se{Introduction}
At large scales spacetime is a pseudo-Riemaniann manifold locally modeled
on Minkowski space. But it is well known that the concurrence of the principles of Quantum
Mechanics and of Classical General Relativity renders this picture
untenable in the small.

Indeed, when we describe the localization of an event by a {\it  point} in a classical manifold we implicitly assume that (in any chart) the coordinates of that point can be simultaneously measured with arbitrarily high precision.
However, high precision in the measurement of at least one coordinate requires, by Heisenberg Uncertainty Principle, the transfer to our observed system of a correspondingly high amount of energy.  If   {\it  at least one} of the coordinates is measured with a high uncertainty, say $L$,
that energy could spread uniformly in a region which, as $L$ increases, becomes infinitely extended in one direction; so that the density of the transferred energy would tend to zero.
Furthermore, if  {\it  all} space uncertainties are kept bounded, and any of the spacetime uncertainties  
is taken smaller and smaller, say of order $l$, the energy transferred by our localization measurement would
increase unlimitedly while remaining concentrated in the fixed bounded region where the event is supposed to be localized. 
If we take into account, together with the principles of Quantum Mechanics, also those of Classical General Relativity, we see that in the second case, for very small values of $l$, a closed trapped surface would be formed, hiding the supposed localization region to any distant observer. If we require that the coordinate uncertainties refer to our  {\it actual observations}, they must be constrained by   {\it uncertainty relations}. 
In the first case, the classical Newton potential generated by the energy transferred by our localization experiment would tend to zero everywhere as $L$ tends to infinity. We cannot expect that neither General Relativistic corrections nor Quantum Gravity effects may be relevant in this case, and we must conclude that {\it a single coordinate can always be measured with arbitrary precision}, provided the precision in the measurement of some other coordinate is sufficiently loose.\footnote{This statement is apparently in conflict with the  GUP expressed by the famous Amati-Ciafaloni-Veneziano uncertainty relations~\cite{ACV}; but the conflict is only apparent,  for the derivations of the ACV relations implicitly assume that the space uncertainty refers to  the value of  {\it all} the space coordinates.} 
 
 In this discussion we adopted the following  {\it Principle of Gravitational
Stability against localization of events}:
 {\bigskip}

{\it The gravitational field generated by the concentration of energy
required by the Heisenberg Uncertainty Principle to localize an event in
spacetime should not be so strong to hide the event itself to any distant
observer - distant compared to the Planck scale.}
{\bigskip}

In \cite{DFR} this principle was used to deduce Spacetime  Uncertainty Relations 
(STUR), where the condition that closed trapped surfaces are not formed was imposed 
using some drastic simplifications: this condition was taken in the simple 
form that $g_ {00}$ should remain positive, where the 
metric   $g_ {\mu\nu}$ was approximated by the solution of the linearized Einstein equations, with source describing the result of a localization experiment in a free field model.
Despite these simplifications, the STUR deduced in \cite{DFR} were seen to 
be compatible with the indications deduced from exact solutions, like 
Schwarzschild or Kerr solutions of Einstein equations. Nevertheless, since 
that time, the need was felt, and often pointed out (see, 
e.g.,~\cite{KARPACZ2001}), of an 
argument free of those sharp approximations.

This concerns especially the linear approximation to the Einstein equations, for the limitations imposed are relevant precisely in the extremely relativistic region where that approximation is no longer valid.
Another important limitation was represented by the use of the notion of  {\it  energy},  through the Heisenberg Principle, which is not defined in a satisfactory way in the General Relativistic context.
In this paper we overcome these two difficulties at the price of limiting ourselves 
to the consideration of spherically symmetric solutions only.\footnote{For an 
independent approach not restricted to spherically symmetric situations, which 
makes appeal to the Hoop Conjecture as a criterion of formation of trapped 
surfaces and to a notion of general relativistic quasi-local energy, 
see~\cite{TV}; the conclusions found there are in agreement with those of 
\cite{DFR} in the case of flat Minkowski space.} Thus, in 
particular, here the space uncertainties take all the same 
value, which, in the special solutions we consider, agrees also with the time 
uncertainty.

Our model, studied in Section~\ref{sec:trapped}, is the theory of a massless quantum field interacting with the {\it classical}  gravitational field, whose source is taken to be the expectation values of the quantum energy momentum tensor in appropriate states.

Our main instruments come from two quite distinct areas:
	
\begin{itemize}
	\item Quantum field theory on curved backgrounds, 
especially in the algebraic formulation~\cite{Wa94, BF}.
   	\item The Raychaudhuri equation for the congruence of null geodesics, providing a rigorous criterion for the formation of trapped surfaces~\cite{HL,Wald}.
	\end{itemize}
	
For completeness, we mention that in the works~\cite{Ch1,Ch1b,Ch2a,Ch2} Christodoulou studied the classical spherical collapse where matter is described by a classical scalar field. In particular, exact solutions of the equations of General Relativity are considered in order to give certain sufficient bounds on the initial values on a complete regular null cone that imply the non-formation of trapped surfaces.
It turns out, however, that to our purposes the full strength of Christodoulou's results is not needed. Following his ideas, here we shall study our quantum matter  and classical gravitational fields on an initial null cone, but the formation of trapped surfaces in the future of such cone will be characterized in Section~\ref{influence} just through the Raychaudhuri equation. We mention, however, that our results are compatible with Christodoulou's ones, see Section~\ref{subsec:comments}.

The combination of the quantum nature of the field with the conditions which prevent the formation of trapped surfaces, yields to 
some constraints on the dimension of the region where measurements can be performed.
Namely we cannot detect anything contained within a certain spatial sphere. The characteristic distance emerging this way is of the order of the Planck length.
The problem of extending these considerations to non-spherically symmetric situations is much harder. Thus we cannot offer neither a refined version of the STUR proposed in \cite{DFR}, derived from exact solutions of the Einstein equations without the use of ill defined notions of energy, nor we can offer accordingly refined versions of the Quantum Conditions on spacetime coordinates proposed there.

However, the Basic Model of Quantum Spacetime of \cite{DFR} can be used as a more 
reasonable geometric background in an approximated approach to the study of 
corrections due to effects of Quantum Gravity. In particular it is interesting to 
study exact spherically symmetric solutions, as so far described, where we take as 
a source for the semiclassical equations the expectation value of the energy 
momentum tensor on Quantum Spacetime. In order to circumvent the problem of not having at our disposal
a version of Quantum Spacetime modeled on a general curved manifold, we adopt the expression of the energy momentum tensor which 
is obtained by generalizing to a curved spacetime the one 
calculated on the Basic Model of Quantum Minkowski Space.

More precisely, in Section~\ref{ref:cosmoapplication}, we shall use this idea to evaluate the influence of the noncommutative spacetime structure in a simple (and thus unrealistic) cosmological model, namely a flat Friedmann-Robertson-Walker spacetime filled with radiation.  To this end, in Section~\ref{subsec:quantumminkowski},
 we take into full account the universal limitation to length scales not smaller than the Planck length in a flat model, by using, in the expression for the energy momentum tensor of our scalar quantum field, the notion of Quantum Wick Product introduced in~\cite{BDFP:2003}. The energy density thus defined is then evaluated in a KMS state 
 describing the background radiation.
Then we obtain a generalization of this result to a flat Friedmann-Robertson-Walker spacetime by exploiting its
conformal isometry with Minkowski spacetime. 
Coupling the energy density so determined with the classical gravitational field, and solving the Einstein equations in the limit of small scale factor, we find that, even if the matter of the model is formed only by radiation (described by a conformal scalar field in a KMS state), the universe obtained in this way shows a phase of power law inflation.
While in the far future no significant modifications to the classical model are obtained, 
close to the initial singularity the Hubble parameter has the behavior $H \sim a^{-1/2}$, i.e.\ it decays much more slowly with respect to the scale factor $a$
 than what would  happen without considering the Quantum Spacetime effects (in which case $H \sim a^{-2}$).
Furthermore, as a byproduct of this analysis in the obtained toy model 
the horizon problem does not arise, see Section~\ref{subsec:quantumFRW}. 
We recall that the horizon problem of standard cosmology arises from the observation that there should be regions of the universe which were never in causal contact since the Big Bang, which seems to be in contrast with the homogeneity of the universe, as shown for instance by the Cosmic Microwave Background. 
Our result here, although obtained under oversimplifying assumptions, 
goes
in the same direction as those drawn at a heuristic level from a full use of the  {\it Principle of Gravitational Stability against localization of events}~\cite{KARPACZ2001, CORFU2006}, which point to a  {\it background dependence of the effective
Planck length}, through which  the a-causal effects which are typical of QFT on noncommutative spacetime may be transmitted over large distances 
(cf comments at the end of Section 3.2).

\section{Formation of trapped surfaces caused by spacetime 
localisation measurements}\label{sec:trapped}

In the present section we shall discuss, by means of the semiclassical Einstein 
equations, the influence of a localized quantum measurement on classical spacetime 
curvature. More precisely, we will consider a massless scalar quantum field $\phi$, 
modeling the measuring device, propagating on a curved spacetime manifold $M$, and 
we will estimate the change of the expectation value of the energy-momentum tensor 
of $\phi$ as a result of the measurement of an appropriate observable localized in a 
bounded region $\cO$ (Proposition~\ref{stima}). Afterwards, we shall use 
this 
estimate to find a sufficient condition on the size of $\cO$ under which a trapped 
surface arises in $M$ due to the backreaction of $\phi$, see Theorem 
\ref{singularity}. This condition expresses only a limitation on the resolution of 
the measuring device. Eventually this condition will be negated according to the 
{\it Principle of Gravitational Stability} and this will provide us with a lower 
bound on the extension of a region in which an event can be operationally 
localized.

Of course, such an analysis of the solutions of the semiclassical Einstein equations, without further restrictions, can be very complicated. For this reason we shall assume that our spacetime manifold $M$ is spherically symmetric with center of symmetry 
described by a world line $\ga$ also contained within $M$. 
For later purposes we shall also assume that $M$ has further standard nice properties like being globally hyperbolic.
Finally, the event we want to detect needs to be spherically symmetric too and furthermore it is taken to be centered around some point of $\ga$.

\subsection{The process of measurement of the 
localisation of an event: model of the quantum detector}\label{subsec:detector}

To begin with, we focus our attention on the preparation of an experiment designed to observe an event localized in a (spherically symmetric) region $\cO$ contained in $M$ and on the quantum effects involved in the process in particular.

A physical procedure which can realize the observation is the scattering process of 
light sent towards a fixed target localized in the region $\cO$. We shall call this 
{\it incoming light}. When the incoming light reaches the region $\cO$, it is 
scattered by the target (the event). Thus the region where the interaction takes 
place can be thought of as being contained within $\cO$. Finally the scattered 
(outgoing) light is measured by some observer localized far away from the target. From the result of the last observation it is in 
principle possible to 
reconstruct the shape of the target and thus detecting an event with some 
precision.

Notice that if we describe light by a quantum field, in the above procedure two 
quantum operations are involved. The first one is needed in order to prepare the 
incoming light in a state in which it is focused towards the target. The second is 
the measurement needed in order to detect the effects of the scattering far away 
from the target. Such an observation is ideally performed by an observer localized 
at future infinity.

We stress that already the incoming field, namely the light sent to the target, 
perturbs the background. If such light is focused too 
much or, equivalently, if 
the target is too small, a trapped surface can occur. In that 
case the scattered light cannot 
reach the observer localized at future infinity. We can thus concentrate ourselves 
on the first part of the measuring process. 
More precisely, we want to evaluate the effect on the curvature of
the focusing of incoming light, in order to
have control on the
formation of trapped surfaces.
Furthermore, both the incoming and outgoing fields 
can be considered as free as
 the interaction with the target is only localized within the region $\cO$.
 Therefore, according to the above observations, it is actually sufficient to assume that the quantum field
we use to model the experiment is a free field.

Let us start formalizing the preceding ideas. As argued above, for our purposes we 
can, for simplicity, describe the light, which is used to 
detect the target localized in 
$\cO$, with a free massless scalar quantum field $\phi$; namely 
a quantum field 
$\phi$ satisfying 
\beq\label{KG} \Box \phi =0\;. \eeq 
We shall quantize such a 
system by means of the algebraic method and therefore we associate to $M$ the 
$*-$algebra $\cA(M)$
 generated by the quantum $\phi(f)$ smeared with compactly supported smooth 
functions $f\in C_0^\infty(M)$. The quantization of such a system is a functorial 
procedure and can be thus completely solved just knowing the geometry of the 
spacetime, see the discussion in \cite{BFV} for a detailed analysis. For our 
purposes, we further suppose that the quantum field $\phi$ is in a quasi free state 
$\om$ (ground state) which is in equilibrium with the background and thus the semiclassical 
Einstein equations 
\beq
\label{eq:initialEinstein} G_{\mu\nu}=8\pi\; \om(T_{\mu\nu}) 
\eeq 
hold (for the definition of the stress-energy tensor $T_{\mu\nu}$ of the field 
$\phi$ we refer the reader to~\cite{BFK,Moretti03}). 

Now let us discuss the state perturbed by the incoming light.
We shall model it by applying $\phi(f)$ to the state $\om$, where $f$ is a real 
valued function supported in the region $\cO$. Notice that, due to the time slice 
axiom \cite{CF}, the region where the preparation of this state is performed can be taken to be any region in the past of $\cO$, 
 provided it contains  $\cO$ in its causal shadow.
 
We indicate by $\om_f$ the quantum state of the 
theory resulting from the above operation; it will be such 
that for every  
$A\in\mA(M)$, 
\beq\label{prepstate} 
\omega_{f}(A) = \frac{\omega \at\phi(f)\; A \;\phi(f)\ct 
}{\omega(\phi(f)\phi(f))}\;. 
\eeq 
The state $\om_f$ can be thought of as the prepared state.
Unfortunately, because of the Reeh-Schlieder theorem, the state $\om_f$ is perturbed everywhere and not only in the region causally connected with the support of $f$.
Actually, if we indicate by 
$(\cH,\pi,\Om)$ the GNS triple corresponding to $\om$,
 a localized perturbation will be realized 
by Weyl operators, namely 
$\Psi = e^{i\phi(f)} \Om\;.$
However, we argue that with a strictly local perturbation generated by Weyl operators, without further restrictions, it is not possible to obtain sensible results.
In fact, the obtained state $\Psi$ is a superposition of states and, in particular, in this superposition, the reference state is always present, actually
$$
(\Om,\Psi) = e^{-\frac{1}{2} \om_2(f,f)}\;.
$$
A more serious drawback lies in the strong continuity of $e^{i\phi(tf)}$
in the real parameter $t$; scaling down $f$, $\Psi$ converges to $\Omega$.
This problem can be avoided by putting restrictions on the energy content
of $\Psi$ \cite{DFR}; here we avoid as much as possible energy
considerations, and follow another route.
In fact, later on, after preparation of such a state, in order to detect particle density we shall use a detector $\cD$ which is normalized on the reference state $\Om$. This means in particular that it will be calibrated to give zero density on $\Om$. It is thus clear that when $\cD$ is tested on $\Psi$ it will give results which are not directly related with the particle density. 
In order to obtain reliable detection we should require the prepared state to be orthogonal to the background, as for example 
$$
\Phi  = \frac{1}{\at 1 - e^{-\mu(f,f)}  \ct} \at e^{i\phi(f)} \Om  - e^{-\frac{\mu(f,f)}{2}}  \Om \ct
$$
where $\mu(f,f)$ is the symmetric part of the two point function of the background state $\om$.
However, such a new state, constructed by means of a linear combination of states, does not enjoy the same nice localization properties of the perturbation as $e^{i\phi(f)}\Om$. Furthermore, the minimum value of the energy transferred to the system with this perturbed state does not differ from the one obtained with the simpler perturbation $\om_f$, Eq.~\eqref{stimaenergia} below.

Finally, we notice that, due to the poorer localization, for fixed total energy the energy density associated to the state~\eqref{prepstate} will be smaller than the one for the strictly localized state 
induced by $\Psi$. This entails that the limitations obtained assuming
~\eqref{prepstate} as a model of localized state will have to be necessarily satisfied by states
with better localization properties.

Here, we are interested in the change of the right hand side 
of~\eqref{eq:initialEinstein} as a result of the observation, and we introduce 
therefore the quantity \beq\label{eq:pertubTmn} \langle T_{\mu\nu}\rangle_{f,0} := 
\om_f ( T_{\mu\nu})-\om ( T_{\mu\nu})\;. \eeq An estimate of this difference is 
given in the following proposition. 

 \begin{proposition} \label{stima} 
Consider a globally hyperbolic spacetime $M$ and the algebra $\cA(M)$ generated by 
the real Klein Gordon field $\phi(f)$ satisfying equation \eqref{KG}. Equip 
$\cA(M)$ with a quasi free Hadamard state $\om$. Then for every real valued 
function $f \in C^\infty_0(M)$, we have the 
following inequality 
\beq\label{stimaenergia} \langle T_{\mu \mu}(x)\rangle_{f,0} 
\geq \frac{1}{2}\frac{|\pa_\mu \Delta(f)(x)|^2}{\om\at \phi(f)\phi(f)\ct} 
\eeq 
where $x$ is a generic point of $M$, $\mu$ is the index of a null direction at $x$ (i.e., $g_{\mu\mu}(x) = 0$) and $\De(f)\in C^\infty(M)$ is the image of $f$ 
under the causal propagator $\De$ on $M$. 
\end{proposition} 

\begin{proof} In 
what follows we shall indicate by $\om_2(f,g)$ the two point function of $\om$ 
which is nothing but $\om(\phi(f)\phi(g))$. Let us start by considering another 
generic real compactly supported smooth function $\xi$. Then since the state $\om$ 
is quasi free, we obtain 
\beq \label{eq:difference} \langle \phi(\xi)\phi(\xi) \rangle_{f,0} := 
\om_f(\phi(\xi)\phi(\xi)) -\om(\phi(\xi)\phi(\xi))= 2\frac{\om_2(f,\xi) 
\om_2(\xi,f)}{\om_2(f,f)} \;. \eeq Since the field $\phi$ is real, we have 
$\om_2(\xi,f) = \overline{\om_2(f,\xi)}$, and hence 
$$ \langle \phi(\xi)\phi(\xi) 
\rangle_{f,0} = 2\frac{ | \om_2(f,\xi) |^2 }{\om_2(f,f)} = 2\frac{ | 
\om_{2,A}(f,\xi) |^2 + |\om_{2,S}(f,\xi) |^2}{\om_2(f,f)} $$ 
where $\om_{2,S}$, 
$\om_{2,A}$ are respectively the symmetric and antisymmetric part of $\om_2$, 
which, at the same time, correspond also to the real and imaginary part of $\om$ 
respectively. Furthermore, the previous chain of equalities ensures that $\langle 
\phi(\xi)\phi(\xi) \rangle_{f,0}$ is a real and positive quantity. We can estimate 
it by noticing that the antisymmetric part of the 
two-point function is given by
$$ \om_{2,A}(f,\xi) = 
\frac{i}{2}\langle\xi,\De(f)\rangle = 
\frac{i}{2}\int_M\xi\De(f)d\mu_g, $$ 
where $\mu_g$ is the volume measure on $M$ 
determined by the metric $g$. Therefore $$ \langle \phi(\xi)\phi(\xi) \rangle_{f,0} 
\geq \frac{1}{2}\frac{|\langle\xi,\De(f)\rangle|^2}{\om_2{(f,f)}}\;. $$ Consider 
now a sequence $(\xi_n)$ of real test functions in $C^\infty_0(M)$ converging 
weakly to $\pa_\mu\delta_x$, the derivative of the delta function supported on the 
point $x \in M$. The thesis of the proposition can then be proven by noticing that, since
$\mu$ is the index of a null direction,
the following limit holds 
$$ \lim_{n\to\infty} \langle \phi(\xi_n)\phi(\xi_n) 
\rangle_{f,0} = \langle T_{\mu\mu}(x) \rangle_{f,0} $$ 
(the limit in the left hand 
side exists thanks to formula~\eqref{eq:difference} and to the hypothesis that 
$\om$ is a Hadamard state) and that $\Delta(f)$ is a smooth function over $M$. 
\end{proof} 

We remark explicitly that the above result does not depend on any 
symmetry assumption about $M$ or $f$. Also, notice that the right hand side
of~\eqref{stimaenergia} can be interpreted as the $\mu$-$\mu$ component of the classical
stress energy tensor associated to the solution $\Delta(f)/\sqrt{\om(\phi(f)\phi(f)}$ of~\eqref{KG}.
This, in turn, can be viewed as the expectation value of the quantum stress energy tensor
in a coherent state as in~\cite{DFR}.

To close this section, we further remark that our inequality~\eqref{stimaenergia} can be seen as a simple case of the more general (null) quantum energy inequalities which have been widely studied in the past, as for example in 
\cite{Ford, Yurtsever, Verch, FR03} (see also~\cite{Fe12} for a comprehensive review).

\subsection{The influence on the curvature and appearance of trapped surfaces}\label{influence}

We are now interested in evaluating the influence of the measuring procedure 
described above on the curvature. In order to estimate the backreaction of the 
observation on the gravitational field we should solve the new equation 
\beq \label{eq:perturbedEinstein}
G_{\mu\nu}=8\pi\; \om_f(T_{\mu\nu}), 
\eeq 
where now $\om_f$ describes the incoming light in the prepared state introduced in 
\eqref{prepstate}. 
We remark that we shall consider $\om_f$ as the state resulting from a measurement performed over the fixed background metric which solves~\eqref{eq:initialEinstein}.
In other words, we shall use \eqref{eq:pertubTmn} to evaluate the effects of a measurement on the expectation value of the stress tensor and we shall consider it as source for gravity in \eqref{eq:perturbedEinstein}.
Precisely at this point we are considering the background metric as fixed.

Solving~\eqref{eq:perturbedEinstein} is of course a rather complicated task, but we can limit ourselves to discussing the backreaction following the localization measurement
merely in terms of conditions 
on the formation of trapped surfaces. In particular, we shall make use of the spherical symmetry in order 
to 
foliate the spacetime by forward pointing spherically symmetric light cones whose 
tips are on a worldline $\ga$. 
We shall select one of these cones $\cC_0$ in such 
a way that the target region $\cO$ (assumed to be open for definiteness) is 
contained in its causal future. Notice that $J^-(\cO)\cap \cC_0$, namely the causal 
shadow of the region $\cO$ on $\cC_0$, cannot be too small and it is controlled by 
the dimension of $\cO$, in the sense that a larger $\cO$ produces a larger shadow 
on $\cC_0$. In particular, we can choose $\cC_0$ arbitrarily close to $\cO$, 
meaning that, for a fixed $\cO$, we can minimize the size of $J^-(\cO)\cap \cC_0$. 
The next step will be to focus our attention on $\cC_0$ and on the Einstein 
equations restricted on it. One of those equations becomes a constraint which 
involves quantities defined intrinsically on $\cC_0$;  this constraint is nothing 
but the Raychaudhuri equation for the congruence of geodesics forming 
$\cC_0$~\cite{Wald}. For our purposes, and in order to extract the minimal length 
we are interested in, it will be enough to consider this constraint.

In order to construct the foliation mentioned above, it is convenient to 
parametrize a normal neighborhood of $M$ containing $\cO$ with the so-called {\bf 
retarded coordinates}. A detailed analysis of such a coordinate system is given for 
example in \cite{Poisson}. Here we shall briefly summarize 
its construction. Let us 
start by recalling that $\ga$ is the world line (a smooth timelike curve) 
describing the evolution of the center of the spatial sphere. We shall parametrize 
the points of $\ga$ by $u:\ga\to\bR$ which is the integral parameter of the forward 
pointing normalized tangent vector field. Let $\cC_u$ be the forward pointing light 
cone formed by all the null geodesics emanating from the point $u$ of $\ga$ and 
traveling towards the future.
The family $\{ \cC_u \}$ foliates the relevant part 
of the manifold $M$, and the target region $\cO$ too.
As a submanifold of $M$, 
$\cC_u$ is topologically $\bR^+\times \bS^2$. Furthermore, any null geodesic 
forming $\cC_u$ is determined by the standard angular coordinates of the unit 
two-sphere $\bS^2$ of the subspace of the tangent space to $M$ in $u$ orthogonal to 
the tangent vector to $\ga$. We shall parametrize such a null geodesic by an affine 
parameter $s$, such that $s$ is equal to $0$ on $u$ and such that the scalar 
product between the tangent vectors in $u$ to the geodesic considered and to 
$\gamma$ is one.
The collection of $s$ for various points on $\gamma$ and for various outgoing directions
forms a scalar field which is usually called {\bf retarded distance}, because it can
 be obtained also as
$s = \left.\partial \sigma/\partial u\right|_{\cC_u}$, where $\sigma$ is the halved squared geodesic 
distance between a point on $M$ and a point $u$ on $\ga$ and $\partial/\partial u$ is applied 
on the second point. 

The most generic spherically symmetric metric respecting this 
structure has the form 
\beq\label{metric} 
ds^2:= - A(u,s) du^2 - 2 ds du + r(u,s)^2  d\bS^2 
\eeq 
where $A(u,s)$ and $r(u,s)$ are spherically symmetric classical fields 
on $M$. Notice 
that the points of $M$ corresponding to a fixed pair $(u,s)$ span a two-sphere 
whose spatial area is equal to $4\pi r^2$. Moreover, where $\pa_s r$ is positive 
$r$ can be used as an alternative coordinate to determine a point on the null 
geodesic.
We stress that at fixed $u$ the relation between the retarded distance $s$ and the radius 
of the sphere $r(u,s)$ can be obtained knowing that $s$ is an affine parameter for the 
null geodesic under investigation.

Let us now come back to the main discussion and let us fix a null cone $\cC_0$ which contains the target region $\cO$ of the measuring process in its causal future $J^+(\cC_0)$. We shall check when  
 a trapped surface arises in $J^+(\cC_0)$ due to the perturbation on the metric induced by the state $\om_f$.
Since we are interested in $J^+(\cC_0)$, we can neglect the backreaction in the past of $\cC_0$.

In our setting, a necessary and sufficient condition for the formation of a trapped surface is the vanishing of 
the expansion parameter $\theta$ of the congruence of null geodesics forming the cones $\cC_u$~\cite{HL}.
A precise definition can be found in the book \cite{Wald}.
Here we need to evaluate the ``change'' of the expansion parameter due to the observation.
The equation governing the evolution of the expansion $\theta$ as a function of the affine geodesic parameter $s$ is the Raychaudhuri equation, namely
\beq\label{Rayc}
\dot{\theta}=-\frac{\theta^2}{2}-R_{ss}
\eeq
and it has to be supplemented by the initial condition
$$
\theta \sim \frac{2}{s} \quad \text{for }s \to 0^+.
$$
In~\eqref{Rayc} $R_{ss}$, assumed to be spherically symmetric, is evaluated at $(s,u)$, and the contributions $\om_{ab}\om^{ab}$ and $-\si_{ab}\si^{ab}$ which are usually present in the Raychaudhuri equation vanish both due to the initial conditions and to the spherical symmetry we have imposed.
The components of the Ricci tensor change due to the
observation and in particular
\beq\label{Rss}
R_{ss}=8\pi\; \om_f(T_{ss}) = 
8\pi \; \om ( T_{ss}) +8\pi \langle  T_{ss}\rangle_{f,0}
= R_{ss}^{(0)}+8\pi \langle  T_{ss}\rangle_{f,0}\,,
\eeq
where $R_{ss}^{(0)}$ is the curvature of the background metric and we have used equation~\eqref{eq:pertubTmn} to evaluate the perturbation induced by the observation on the state $\om$ over the fixed background metric. 
We have furthermore used the fact that, according to~\eqref{metric}, $g_{ss}=0$. 

We are now ready to introduce the main theorem of the present section which states that, under fairly general assumptions on the original state $\om$ and on $M$ (see the discussion in the next subsection), if 
the initial data for the matter are supported in a sufficiently small region on $\cC_0$, a trapped surface
arises in the future of $\cC_0$.
We stress once again that the initial data on $\cC_0$ mentioned in the theorem need to be interpreted as arising because of the localization experiment considered, namely because of the ``incoming light'' we are sending towards the target contained
in $\cO$.

\begin{theorem}\label{singularity}
Consider a spherically symmetric spacetime $M$ parametrized with the retarded coordinates and  the $*-$algebra $\cA(M)$ of
the massless minimally coupled free Klein Gordon field on $M$.
Let $\om$ be a quasi-free Hadamard state for $\cA(M)$ such that:
\begin{enumerate}
\item\label{it:semiclass} the semiclassical Einstein equations are satisfied by $\om$ and $M$;
\item\label{it:positive} $R^{(0)}_{ss}= 8\pi \om(T_{ss})$ is positive on $\cC_0$;
\item\label{it:continuity} there is a constant $C > 0$ such that the two-point 
function $\om_2$ of $\om$ fulfills, for every $f$ 
supported in $J^+(\cC_0)$,
\beq\label{continuity}
|\om_2({f,f})| \leq C \| s \psi_f\|_2 \| \pa_s (s\psi_f) \|_2 ,
\eeq
where now $\psi_f$ is equal to $\Delta(f)$, computed with respect of the unperturbed metric and restricted on ${\cC_0}$, and where $\|\cdot\|_2$ is the $L^2$ norm on $\cC_0$ with respect to the product measure of $ds$ with $d\bS^2$,  the standard measure of the unit two dimensional sphere. 
\end{enumerate}
Assume now that $f$ is a spherically symmetric function   
supported in a region $\cO\subset J^+(\cC_0)$, chosen in such a way that
$J^-(\cO)\cap \cC_0$ is formed by points in the past of a sphere of $\cC_0$ determined by the equation $s =s_2$.
Furthermore, suppose that there is an $s_1$, with $s_1 < s_2 < \frac{3}{2}s_1$,
such that
\beq\label{constraint}
\|\pa_s\psi_f\|^2_2\leq 8\pi\int_{s_1}^{s_2} |\pa_s\psi_f|^2 ds\,.
\eeq
Then there is a constant $\bar s := 1/\sqrt{12 C}$ such that if $s_2 < \bar s $
the expansion parameter $\theta$ of the congruence of null geodesics defining $\cC_0$ for the metric satisfying \eqref{eq:perturbedEinstein} vanishes on a sphere contained in $\cC_0$, and thus $J^+(\cC_0)$ contains a trapped surface.
\end{theorem}

\begin{proof}
Let us start by writing the equation \eqref{Rayc} governing the evolution of the expansion parameter $\theta$ on $\cC_0$ in integral form:
$$
\theta(s_2)=\theta(s_1)-\int_{s_1}^{s_2} \frac{\theta^2}{2} ds-\int_{s_1}^{s_2} R_{ss} ds\,.
$$
We get immediately the following inequality
$$
\theta(s_2) \leq \theta(s_1)-\int_{s_1}^{s_2} R_{ss} ds\,.
$$
We can now use the expansion \eqref{Rss} and the fact that $R^{(0)}_{ss}$ 
is positive on $\cC_0$ (hypothesis~\ref{it:positive}) to write
$$
s_2\theta(s_2) \leq s_2\theta(s_1)
- 
 8\pi s_2\;\int_{s_1}^{s_2} \langle T_{ss} \rangle_{f,0}\;  ds\,.
$$
Rewriting the Raychaudhuri equation as
$$\frac{d}{ds}\frac{1}{\theta} = \frac{1}{2} + \frac{R_{ss}}{\theta^2},$$
from the initial condition for $\theta$ on $\ga$  and from the fact that $R_{ss} = R_{ss}^{(0)} + \langle T_{ss}\rangle_{f,0} \geq 0$ on $\cC_0$,
 we deduce that $\theta(s_1) \leq  2/{s_1}$. Thus, using the estimate given in \eqref{stimaenergia} and the continuity enjoyed by $\om_2$ on $\cC_0$ (hypothesis~\ref{it:continuity}) we have 
\begin{equation}\label{eq:inequality1}
s_2\theta(s_2) \leq 
 2 \frac{s_2}{s_1}
- 
\frac{s_2}{C}  \frac{4\pi\int_{s_1}^{s_2}|\pa_s\psi_f |^2ds}{\|s\psi_f \|_2\|\pa_s(s\psi_f)\|_2} \leq   2 \frac{s_2}{s_1}
- 
\frac{s_2}{2C}  \frac{\|\pa_s\psi_f \|^2_2}{\|s\psi_f \|_2\|\pa_s(s\psi_f)\|_2}  \;,
\end{equation}
where in the last inequality we have used \eqref{constraint}.
We now notice that, being $\psi_f$ smooth, there exists $\bar s \in (0,s_2)$ such that $\|\psi_f\|_\infty = |\psi_f(\bar s)|$, and therefore, again thanks to the support properties of $\psi_f$,
\begin{equation*}
\|\psi_f\|_2^2 \leq 4\pi s_2 |\psi_f(\bar s)|^2 = 4\pi s_2 \left|\int_{\bar s}^{s_2}\pa_s\psi_f ds\right|^2 \leq s_2^2 \|\pa_s\psi_f\|_2^2,
\end{equation*}
having used Cauchy-Schwarz in the last inequality.
Together with the inequality $\|s \psi_f \|_2 \leq s_2 \| \psi_f \|_2$ this implies
$\|\pa_s (s \psi_f) \|_2 \leq 2 s_2 \| \pa_s\psi_f \|_2$. Inserting then everything into equation~\eqref{eq:inequality1} and
recalling that by hypothesis $s_2/s_1 < 3/2$, we get
$$
s_2\theta(s_2) \leq 3 - \frac{1 }{4 C\,(s_2)^2} \;.
$$
Notice that from the last inequality we have that $\theta(s_2)$ is surely negative if $ (s_2)^2 < 1/(12 C), $
which is one of the 
hypothesis of the theorem. Furthermore, if the expansion is negative on $\cC_0$ at $s$ it remains negative also for every point of $
\cC_0$ in the future of $s$ and thus a trapped surface forms in $J^+{(\cC_0)}$. 
\end{proof}
Before proceeding with our discussion we briefly comment on 
the constraint~\eqref{constraint} imposed on the $L^2$-norm of $\pa_s\psi_f$.
Notice that, when the past directed null geodesics emanated from $\cO$ meet $\cC_0$ in the region determined by two constants $s_1$ and $s_2$ as in the theorem, the singularities of the causal propagator $\De$ restricted on $\cO\times \cC_0$ are contained within that region.
Since the dominant contribution to $\|\pa_s\psi_f\|_2$ comes from such singularities, one can expect that, in this situation, equation~\eqref{constraint} is satisfied.

We see therefore that if the incoming light is focused too much, namely when $\cO$ is too small, a trapped surface occurs.
Furthermore, up to some mild hypotheses, such a condition is independent on the shape of the incoming light, only the resolution of the detector is important.

Imposing now the Principle of Gravitational Stability against localization of events we conclude that the hypotheses of Theorem \ref{singularity} must be violated, and this implies, if $C \simeq 1$ (see appendix~\ref{Restriction} and the discussion in the next subsection), that the region $\cO$ containing the support of $f$, in which the event to be observed is localized, has to be at least of the size of the Planck length. 
This result generalizes to a curved setting the 
particular case of the 
STUR of~\cite{DFR} in which all the uncertainties are of the same size. In order to get a generalization of the full set of STUR it would be necessary to extend the previous analysis to the non-spherically symmetric case, a task which is beyond the scope of the present work.

\subsection{Comments}\label{subsec:comments}
Let us now briefly discuss the hypotheses adopted in 
the previous theorem.

First of all, let us recall that by the hypotheses~\ref{it:semiclass} and \ref{it:positive} of Theorem \ref{singularity} the unperturbed spacetime
satisfies a semiclassical Einstein equation and the $R^{(0)}_{ss}$ component
of the Ricci tensor is positive. 
Notice that it is possible to provide semiclassical models of quantum fields interacting with gravitation where both facts are satisfied
 at least when the background is a flat Robertson-Walker spacetime, namely when there is a single dynamical degree of freedom, the cosmological scale factor $a(t)$ given in terms of the cosmological time $t$. 
In fact, recently it has been proven that exact solutions of such semiclassical system  
 exist \cite{n}, and it is thus meaningful to assume that in the unperturbed background $R^{(0)}_{ss}=8\pi \om(T_{ss})$.
 Furthermore, in a Robertson-Walker spacetime
$$
R^{(0)}_{ss} = -2\frac{\dot {H}}{a^2},
$$
where $H$ is the Hubble parameter, i.e.\ the logarithmic derivative of the scale factor with respect to the cosmological time and where dot stands for the derivative in the cosmological time. 
Notice that in an expanding universe like the one in which we are living $\dot{H}$ is negative (as can be seen analyzing the equation of state of the universe). It is thus reasonable to assume $R_{ss}^{(0)}$ to be a positive quantity.

Hypothesis~\ref{it:continuity} in Theorem \ref{singularity} could appear as a strong requirement about the continuity properties enjoyed by the background state $\om$. We would like to stress that the continuity condition 
\eqref{continuity} is satisfied by the massless vacuum on a Minkowski background~\cite{BM} (see also  the appendix).
Furthermore, on a curved spacetime it is possible to construct states that have similar properties \cite{DPP}.

Of course, a different choice of normalization of the background state could alter somewhat the results we obtained here. For example it can be shown that a slightly more stringent choice, like 
\begin{equation}\label{strongercontinuity}
|\om_2({f,f})| \leq C \| r \psi_f\|_2 \| \pa_s (r\psi_f) \|_2 ,
\end{equation}
would result in a limitation for the radius $r$ of the area of minimal localization, similar to that obtained above for the affine parameter $s$.
In any case the two choices lead to results that agree at the first order, as $r(u,s) \sim s$ for $s \to 0$. 

It seems interesting to notice that both inequalities 
satisfied by $s_1$ and $s_2$ stated in the hypotheses of Theorem \ref{singularity}
are sufficient conditions for the formation of a singularity in the case of
 classical spherical collapse of matter described by a scalar fields,
 provided the support of $\psi_f$ is contained within $(s_1,s_2)\times\bS^2$.
 In fact, in a series of papers \cite{Ch1,Ch1b,Ch2a,Ch2} Christodoulou has studied the spherical collapse induced by a classical matter which is described by a scalar field minimally coupled with the curvature.
In those papers the equations governing the dynamics of the coupled matter-gravity system are cast in a form such that initial values for the gravity and matter fields are given on an initial surface that has the shape of a null cone ($\cC_0$ in the notation introduced above).

In \cite{Ch1} Christodoulou has given a 
condition for 
the initial values of $\phi$ on 
the cone $\cC_0$ which guarantees that no black hole forms in the future of $\cC_0$ and hence, that no trapped surfaces is contained within  $J^+(\cC_0)$. 
This condition essentially requires that the initial values of $\phi$ do not vary too much on $\cC_0$.

In a subsequent paper \cite{Ch2a} a condition on the initial data which guarantees the formation of trapped surfaces and hence of singularities is also given, which, in the geometric framework introduced in Theorem \ref{singularity}, can be expressed in terms of the radiuses  $r_1, r_2$ and the Hawking masses $m_1,m_2$ of two spheres $S_1$ and $S_2$ contained in $\cC_0$.
More precisely, there are two positive constants $c_0$ and $c_1$ such that if both
$$
\frac{r_2}{r_1}-1 \leq c_0 \;\qquad \text{ and } \qquad \frac{2\at m_2-m_1\ct}{r_2} \geq - c_1 \at\frac{r_2}{r_1}-1\ct \log \at
 {\frac{r_2}{r_1}-1} \ct $$ are satisfied, the future of $\cC_0$ contains 
a spacelike singularity. It turns out that these requirements
are both satisfied if $s_1$  and $s_2$ obey the inequalities stated in
Theorem \ref{singularity}.

\section{An application of Quantum Spacetime in a simple cosmological model}\label{ref:cosmoapplication}
The information we can extract from Theorem \ref{singularity} is that, 
if we want to avoid that the localization of an event causes the 
formation of a singularity, the localization of that event should not take 
place in small regions of spacetime.
At first order, an estimate of the size of those regions is given by the quantity $\overline{s}$ in Theorem 
\ref{singularity},
which is a length of the order of the Planck scale. 
Later on we shall be more precise on this point for the class of space-times we are going to describe.
In any case, on the basis of our comments in the preceding section regarding Eq.~\eqref{strongercontinuity},  we shall bound our analysis to the first order in $s$.

The natural scenario where these ideas can be tested is the cosmological one, where there is spherical symmetry with respect to every point, and thus the result obtained in the preceding section can be considered.
Of course, since the minimal length scale is constant in time, it is expected that the effects due to the noncommutative nature of spacetime become  important in the past, namely close to the Big Bang when the scale factor was comparable with such length scale and when 
the universe was very dense and very hot.
We are thus interested in understanding how the back reaction of matter in thermal states on curvature is modified by the introduction of the minimal length scale. 
The task we are facing is therefore to compute the modification of the energy density of quantum fields propagating on a cosmological spacetime due to the sharpest localization. 
We shall accomplish this task discussing the influence of the minimal length scale, obtained in the previous section, in a toy model consisting of an universe filled only with radiation. We are particularly interested in this influence close to the initial singularity which we call Big Bang.   

Our starting point will be the remark that the limitations to localizability have an effect in the evaluation of product of fields at the same point, like $\phi^2(x)$, 
which are objects appearing in the definition of the energy density. On Minkowski spacetime this was implemented, in~\cite{BDFP:2003}, by considering product of fields at different points, 
say $x$ and $y$, and using states of optimal localization
on the model of Quantum Spacetime of~\cite{DFR} in order to minimize the difference
$x-y$ in a way compatible with the commutation relations.

For our purposes it would be necessary to repeat the analysis presented in \cite{BDFP:2003} on a curved spacetime $M$.
This would require the introduction of a full set of commutation relations between coordinates on $M$, defining a noncommutative algebra $\sE$ which replaces the algebra of smooth functions on $M$, and of a full-fledged quantum field theory on $\sE$. However, because of the lack of a thorough analysis of the operational limitations
to localizability in more general (i.e.\ non-spherically symmetric) spacetimes, we can not commit ourselves
to a specific choice of commutation relations. Therefore, we adopt the strategy
of analyzing the effect of noncommutativity on the energy density of a quantum field first on Minkowski
spacetime, where the algebra $\sE$ is well known~\cite{DFR}, and then we will discuss a possible
extension of the result thus obtained to a flat Friedmann-Robertson-Walker (FRW) spacetime, namely a spacetime $M$ where the metric is
\beq\label{FRW}
ds^2=-dt^2+a(t)^2\aq dx_1^2+dx_2^2+dx_3^2,\cq
\eeq
where $a(t)$ is the scale factor and $t$ is the cosmological time.

Then this energy density will be used as source in the Friedmann equation in order to estimate its global effect on the curvature. The result is that, although the Big Bang singularity is still present in the past of the model, the 
scaling behavior of radiation density close to the singularity is significantly modified.
In this way, the resulting spacetime appears as a power law inflationary scenario.
Furthermore, because of this modification the initial singularity is represented by a lightlike surface.
Thus in such a spacetime every couple of points have been in causal contact at some time after the Big Bang, 
and hence, the horizon problem of the standard cosmological model is avoided.
We also mention that a similar result was obtained in~\cite{DFP}, as a consequence of a consistent renormalization of the energy density of the quantum field on a commutative spacetime.

\subsection{Energy density on Quantum Minkowski Spacetime}\label{subsec:quantumminkowski}

The C$^*$-algebra $\sE$ of Quantum Spacetime is generated by self-adjoint operators $q^\m, Q^{\m\n}$ subject to the following relations, in which $\la$ stands for the Planck length,
\begin{equation}\begin{split}\label{eq:qc}
[q^\m,q^\n] = i\la^2 Q^{\m\n}, \qquad [q^\rho, Q^{\m\n}]= 0,\\
Q_{\m\n} Q^{\m\n}=0,\qquad \Big(\frac{1}{4}Q^{\m\n}(*Q)_{\m\n}\Big)^2 = 1,
\end{split}\end{equation}
and given a suitable function $f$ on $M$ (e.g.\ $f\in C_0^\infty(M)$) one defines an element $f(q) \in \sE$ by
\begin{equation}\label{eq:fq}
f(q) := \int d^4 k \ahat{f}(k) e^{ikq},
\end{equation}
where $\ahat{f}(k) = \int \frac{d^4 x}{(2\pi)^4} f(x) e^{-ikx}$ is the inverse Fourier transform of $f$.

Consider now a free massless scalar field $\phi$ propagating on Minkowski 
spacetime, equipped with the standard Minkowskian metric 
$$ 
ds^2= - dt^2 + dx_1^2+dx_2^2+dx_3^2.
$$ 
Since the stress-energy tensor is built from 
products of (derivatives of) $\phi$ evaluated at the same spacetime point, 
we start by recalling the definition of the quantum diagonal map 
of~\cite{BDFP:2003}, which generalizes to the Quantum Spacetime the map of 
evaluation at coinciding points of a function $f(x,y)$ of two commuting 
variables. To this end, we consider the tensor product algebra $\sE^{(2)} 
= \sE \otimes_\sZ \sE$, $\sZ$ being the center of $\sE$, and we introduce 
the operators \begin{equation*} q_1^\mu := q^\mu \otimes 1, \qquad q_2^\mu 
:= 1\otimes q^\mu, \end{equation*} describing the quantum coordinates of two independent events.
Note that adopting the ``$\sZ$-bimodule tensor product'', rather than the 
usual tensor product over complex numbers, amounts to requiring that the 
commutator of the different components of the coordinates is the same for 
all events, that is  
\begin{equation}
[q_{2}^\m,q_{2}^\n] - [q_{1}^\m,q_{1}^\n] = i\la^2 ( 1\otimes Q^{\m\n} -  
Q^{\m\n}\otimes 1)  =  i\la^2  dQ^{\m\n} =  0
\end{equation}
(cf.~\cite{BDFP:2003, BDFP:2010}). Introducing furthermore the 
center of 
mass and relative distance coordinates 
\begin{equation}\label{eq;CenterMassRelative} \bar q^\mu := 
\frac{1}{2}(q_1^\mu + q_2^\mu), \qquad \xi^\mu := 
\frac{1}{\la}(q_1^\mu-q_2^\mu), \end{equation} the quantum diagonal map is 
given by the conditional expectation $E^{(2)} : \sE^{(2)} \to \bar \sE := 
C^*(\{ e^{i k\bar q}\,:\,k \in \bR^4\})$ defined, on a generic element 
$f(q_1,q_2) \in \sE^{(2)}$ by \begin{equation}\label{eq:conditional} 
E^{(2)}(f(q_1, q_2)) =\int d^4k_1 d^4k_2\, \ahat f(k_1, k_2)\ 
e^{-\frac{\la^2}{4}|k_1-k_2|^2}e^{i(k_1+k_2)\bar q}, \end{equation} where 
$|k|^2 = \sum_{\m=0}^3 k_\m^2$ is the squared Euclidean length of $k \in \bR^4$.

Using $E^{(2)}$
we will define the \emph{quantum
Wick square} of $\phi$ as
$$
:\phi^2:_Q(\bar q) := E^{(2)}(:\phi(q_1)\phi(q_2):). $$

At the same time the energy density is defined recalling the form of the $00$ component of the stress tensor, namely $T(\pa_t,\pa_t)$:
\beq\label{eq:density}
:\rho:_Q(\bar q) :=  E^{(2)}\Big(:\pa_0\phi(q_1)\pa_0\phi(q_2): -\frac{1}{2}\eta_{\mu\nu} :\pa^\mu\phi(q_1)\pa^\nu\phi(q_2):  \Big)\;.
\eeq
We shall now discuss the expectation value of these observables in suitable states.
More precisely we will consider, on the free field algebra, the KMS state
$\om_\beta$ at inverse temperature $\beta$, whose two point function is given by
\begin{align}\label{eq:KMS}
\om_\be(\ahat\phi(k_1)\ahat\phi(k_2)) &=\frac{1}{(2\pi)^3} \de(k_1+k_2)\de(k_1^2)\frac{\varepsilon(k_1^0)}{1-e^{-\be k^0_1}}.
\end{align}
The two point function of the vacuum state $\om_0$ can be obtained by considering the zero temperature limit of the previous expression.
In order to analyze the effects of the noncommutativity of spacetime on the expectation values of the Wick square $\phi^2$ and the energy density $\rho$, we evaluate the renormalized versions of these observables. Thus taking into account formulas~\eqref{eq:conditional}, \eqref{eq:KMS} we get
\begin{equation}\label{omegabeta}\begin{split}
\om_\be(:\phi^2:_Q(\bar q)) := (\om_\be-\om_0)(\phi^2_Q(\bar q)) &= \frac{1}{(2\pi)^3} \int d^4k \,\de(k^2)\frac{e^{-\frac{\la^2}{2}|k|^2}}{e^{\be|k^0|}-1} =\frac{1}{2\pi^2} \int_0^{+\infty}dk\,k \frac{e^{-\la^2k^2}}{e^{\be k}-1},
\end{split}\end{equation}
while for the renormalized energy density, defined in the same way as above, we get 
\begin{equation}\begin{split}\label{eq:quantumrho}
\om_\be(:\rho:_Q(\bar q))  &=\frac{1}{2\pi^2} \int_0^{+\infty}dk\,k^3  \frac{e^{-\la^2k^2}}{e^{\be k}-1}.
\end{split}\end{equation}
For our purposes it is important to pinpoint both the asymptotic form of $\rho$ for small and large $\la/\beta$:
$$
\om_\be(:\rho:_Q(\bar q)) \simeq C_1\frac{1-\frac{\la^2}{\beta^2} C_2}{\beta^4} \;,\quad \frac{\la}{\beta}  << 1 \qquad\text{while}\qquad
\om_\be(:\rho:_Q(\bar q)) \simeq C_3 \frac{1}{\beta \la^3} \;,\quad \frac{\la}{\beta}  >> 1\;,
$$ 
where $C_1$, $C_2$ and $C_3$ are fixed constants
\footnote{The numerical values of these constants are 
$C_1=\frac{\pi^2}{30}$,
$C_2= \frac{40}{21}\pi^2$ and 
$C_3=\frac{1}{8\pi^{3/2}}$.
}.

Notice that while for small $\la/\beta$ the effects of the noncommutativity of 
spacetime can be considered as a small correction, for large values of 
$\la/\beta$ the form of $\om_\be(:\rho:_Q(\bar q))$ appears to be completely 
different from its classical (i.e.\ commutative) counterpart.

\subsection{Backreaction on Quantum (FRW) Spacetime}\label{subsec:quantumFRW}

As discussed at the beginning of this section, we are now interested in solving a 
semiclassical Einstein equation where the effects of the non\-commutativity of 
spacetime are taken into account in the evaluation of the matter stress tensor, 
while the curvature is treated classically. In other words, the 
semiclassical 
Einstein Equations take the form
\beq\label{sem}
G_{\mu\nu}=8\pi \om (:T_{\mu\nu}:_Q).
\eeq
For simplicity, we shall further assume that the metric is of the form \eqref{FRW} and we shall consider the matter to be described only by a conformally coupled massless scalar field.
Hence, thanks to the large spatial symmetry, the equation is equivalent to the first Friedmann equation which looks like
\beq\label{Friedmann}
H^2(t)=\om(:\rho:_Q(\bar q)).
\eeq
In \eqref{sem} and \eqref{Friedmann} we have considered the matter in a suitable quantum state $\om$ and we have used the renormalized 
stress tensor or energy density
because we require that the limit $\la\to0$ should be equivalent to the semiclassical Einstein equation on classical spacetime
(recall that $\la$ is the parameter measuring the noncommutativity of spacetime). 
In other words, when $\la$ is very small, the expectation values of both stress tensor and the energy density are required to be finite.

Since, here, we are considering conformal matter, it is meaningful to chose the conformal vacuum as the reference state. At the same time, since the state describing today ordinary matter should be a thermal state with very low temperature, we shall take as the state $\om$ appearing in~\eqref{Friedmann} a conformal KMS state that we shall indicate by $\om^M_\beta$. 
We shall thus try to obtain the expectation value of the energy density in such a state generalizing the discussion given in the preceding section to a curved spacetime.
 
In order to do that, first we would like to give an estimate of the minimal length scale $\la$ under which localization cannot be achieved on FRW spacetime, and to check its dependence on time in particular. To this end, we shall specialize the construction of the retarded coordinates performed in Section \ref{influence} with respect to the worldline $\ga = \{ (t,0,0,0)\in M \,:\, t > 0\}$. It is straightforward to verify that the flat FRW metric (in spherical spatial coordinates) $ds^2 = -dt^2 + a(t)^2[dr^2+r^2d\bS^2]$, is reduced to the form~\eqref{metric}
through the change of variables
\begin{align}
u &= t(\tau(t)-r), \label{eq:changeu}\\
s &= \frac{1}{a(u)}\int_0^r \tilde a(\tau(u)+r')^2dr', \label{eq:changes}
\end{align}
where $\tau\mapsto t(\tau)$ is the change of coordinates which maps the conformal time $\tau$ to the cosmological time $t$, $t \mapsto \tau(t)$ its inverse, and where $\tilde a(\tau) := a(t(\tau))$.

We now  recall the results stated in Theorem \ref{singularity}, which says that we cannot localize objects in a small region $\cO$ contained within a given null cone $\cC_t$, described, in retarded coordinates, by the equation $u = t$. Actually the set of points of $\cC_t$ in causal contact with $\cO$ cannot be contained within the set  $\cR_t \subset \cC_t$ determined by the relation $s \leq \overline{s}$, where $\overline{s}$ is a constant that does not depend on time. From such a statement it is possible to estimate the minimal detectable length scale at any time by measuring the size of the set of points at fixed cosmological time $t$ which are in causal contact with $\cR_t$, namely the size of the region 
$J^-(\cR_t) \cap \Si_t$, where $\Si_t := \{(t',x_1,x_2,x_3)\in M \,:\, t'=t \}$.
Since we expect the latter to be very small, of the order of the Planck length, the corresponding coordinate $r$ will also be very small, at least if $a(t)$ is not too small too. In this approximation, eq.~\eqref{eq:changes} can be replaced by $s = a(u)r$, so that the $r$ coordinate of the upper border of $\cR_t$ is given by $\overline{r}(t) = \overline{s}/a(t)$. Therefore we see that the size of the region $J^-(\cR_t) \cap \Si_t$ can be estimated by $\la = 2 a(t)\overline{r}(t) = 2 \overline{s}$, and it is constant in time within our approximations. 
Later on we shall check that the results we are going to derive are consistent with this approximation.
To be precise, the approximated result can be made exact using the slightly more stringent continuity condition discussed
in \eqref{strongercontinuity}.

We will therefore assume from now on that the minimal localization length $\la$ is constant in time. Furthermore we will use the state $\om_{\beta_0}$ defined in \eqref{omegabeta} to construct
the corresponding state $\om^M_\beta$ for the conformally coupled massless scalar field in $M$ (via the pullback with respect
to the conformal isometry mapping $M$ into Minkowski space~\cite{PiConf}).
We notice that, since the considered field is conformally invariant,  the new state on Friedmann-Robertson-Walker spacetime 
appears as a conformal KMS state, namely it appears as a state which enjoys the KMS condition with respect to conformal 
time translations, which represent an accelerated observer. 
In this case the inverse temperature is $\beta_0$ while the 
physical temperature, often called Tolman inverse temperature\footnote{We refer the reader to \cite{Dixon} for the definition of the relativistic temperature.} needs to be rescaled by $a(t)$ and it is $\beta(t)=\beta_0 a(t)$.
On the other hand, if we use the point of view of Buchholz, Ojima and Roos \cite{BOR}, and we evaluate the temperature 
in $M$ using the expectation value of $\phi^2$ as a thermometer, we obtain that 
the physical inverse temperature $\beta(t)$ scales as $\beta_0 a(t)$, being $a(0)=1$, in agreement with the Tolman one. It must be stressed that the use of a thermal equilibrium state at early epochs of the Universe evolution is clearly an approximation, since the state of the Universe was quite far from equilibrium at those epochs; yet an approximation of this kind is commonly accepted, in view of the fact that for very small values of the scale factor thermodynamic equilibrium was easily established.

Therefore, in order to take into account the effect of the quantum nature of spacetime, we shall assume that in passing from Quantum Spacetime modeled on Minkowski spacetime to a Quantum Spacetime modeled on a FRW one, the only effective 
change  on the quantum field energy density is given by substituting $\beta$ in equation~\eqref{eq:quantumrho} with $\beta(t) = \beta_0a(t)$.
Thus we get, for the energy density $\rho_\beta$ in such a noncommutative spacetime,
$$
\rho_\beta(t):=\om^M_\be(:\rho:_Q(\bar q)) = \frac{1}{2\pi^2} \int_0^{+\infty}dk\,k^3  \frac{e^{-\la^2k^2}}{e^{\be(t) k}-1}.
$$

Suppose now to have an eternally expanding universe, and let us discuss the form of the energy density.
In the future, at a certain point, we will have that 
$
{\la}/{\beta}={\la}/{( \beta_0 a(t)) }
$
becomes much smaller than $1$, and in that regime the energy density looks like
$$
\rho_\beta(t)\simeq C_1\frac{1}{\beta_0^4 a(t)^4} \at 1-\frac{\la^2}{\beta_0^2 a(t)^2} C_2 \ct.
$$
Therefore the effect due to the noncommutativity of the underlying spacetime is very small and can actually be neglected.
On the contrary, when the universe was very small (close to the Big Bang) we have that $\la/\beta$ was much bigger than 1
and in that regime the energy density scales with respect to $a$ as
$$
\rho_\beta(t)\simeq  \frac{C_3}{\beta_0 a(t) \la^3}.
$$ 
Notice that this is less divergent, for $a\to 0$, than on the corresponding classical spacetime.

We have now all the ingredients in order to evaluate the backreaction close to the Big Bang, namely we can solve approximatively 
the semiclassical Friedmann equation \eqref{Friedmann}
which, thanks to the preceding discussion, takes the following 
very simple form in the limit of small $a$:
\beq\label{Hbigbang}
H^2(t)=  \frac{c}{a(t)}\;,
\eeq
where $c=(8\pi C_3)/(3\beta_0\la^3) $.
The solution of this equation is very simple too. It corresponds to a Power Law inflationary scenario, and thus to a spacetime which does not present the horizon problem, as can be seen by observing that the conformal time tends to $-\infty$ on the solutions of the previous differential equation when $a(t)$ tends to vanish in the past.
In order to make the last point clear, we shall conformally embed the cosmological spacetime in a Minkowski one and we shall 
analyze the form of the initial singularity therein.
Let us thus study the value of the conformal time close to the Big Bang singularity.
In an eternally expanding universe, the conformal time $\tau$ in the past of  a fixed cosmological time $t_0$ is 
given by
\begin{equation}\label{eq:powerlaw}
\tau_0-\tau=  \int_{t}^{t_0}   \frac{1}{a(t')}   dt'  =  \int_{a}^{a_0}   \frac{1}{{a'}^2 H(a')}   da' = \frac{1}{\sqrt{c}} \int_{a}^{a_0}   \frac{1}{{a'}^{3/2}}   da' = \frac{2}{\sqrt{c}} \at \frac{1}{\sqrt{a}} - \frac{1}{\sqrt{a_0}} \ct 
\end{equation}
where we have used $a$ as a time measure and the explicit expression of the Hubble constant $H$ obtained in \eqref{Hbigbang} close to the Big Bang.
Out of the preceding result we notice that at the Big Bang, namely when the scale length $a$ vanishes, the conformal time $\tau$ tends to $-\infty$.
Hence, the singularity is located at the past boundary of the conformally related Minkowski spacetime which is a lightlike singularity hypersurface.
Thus in this spacetime every couple of points have been in causal contact at some time in the past after the Big Bang, avoiding the horizon problem 
present in the standard cosmological models.

The preceding result has been derived assuming $\la$ constant in time. 
That it is actually consistent with such an assumption can be seen by
noticing that inserting~\eqref{eq:powerlaw} into eq.~\eqref{eq:changes}, recalling that $u = t$ on the null cone $\cC_t$, and setting $s = \overline{s}$, $r = \overline{r}$ entails
$$
\overline{s} =\frac{c}{3t^2}\at \frac{1}{\tau(t)^3}  - \frac{1}{(\tau(t)+\overline{r})^3} \ct,
$$
whose leading behavior, for $\tau$ that tends to $-\infty$, is $\overline{s} = \overline{r}/\tau^2\big(1+O(\overline{r}/\tau)\big) = \tilde a(\tau) \overline{r}\big(1+O(\overline{r}/\tau)\big)$.
Thus the assumption that $\la=2a(t) \overline{r}(t)$ holds also close to the Big Bang.

Our conclusions, eq.~\eqref{eq:powerlaw}, agree with the heuristic argument~\cite{KARPACZ2001, CORFU2006} which suggests to modify the Planck length in~\eqref{eq:qc} by, as a rough approximation,  the factor $g_{00}^{-1/2}$; the  minimal distance between two events \cite{BDFP:2010} would then be modified accordingly. Such a rough argument points too to an infinite extension of non local effects near a singularity, where  $g_{00}$ vanishes; so that, near the ``Big Bang",  thermal equilibrium would have been  established globally.

We would like to conclude this section with a more heuristic argument
which supports the results we presented here above.
In the case of spherically symmetric background and of localization of an
event with the same symmetry, an argument which is not based on the linear approximation (but is still 
obviously heuristic) can be outlined as follows.

Suppose that our background state describes the distribution of the total
energy $E$ within a sphere of radius $R$, with $E  <  R$. If we localize,
in a spherically symmetric way, an event at the origin with space accuracy
$a$, due to the Heisenberg Principle the total energy will be of the order
$1/a  +  E$. We must then have
$$
\frac{1}{a}  +  E   <   R,
$$
otherwise our event will be hidden to an observer located far away, out of
the sphere of radius $R$ around the origin.
Thus, if $R  -  E$ is much smaller than $1$, the ``minimal distance" will
be much larger than $1$.
But if $a$ is anyway larger than $R$ the condition implies rather
$$
\frac{1}{a}  +  E   <   a.
$$
Thus, if $R  -  E$ is very small compared to $1$ and $R$ is much larger
than $1$, $a$ cannot be essentially smaller than $R$.
This naive picture suggests too that, due to the principle of
Gravitational Stability,  initially all points of the Universe should have been
causally connected.

\section{Final comments and outlook}

In this paper we have analyzed some bounds on the quantum nature of spacetime 
assuming the \emph{Principle of Gravitational Stability against localization of 
events}, i.e.\ that by just observing 
the localization of an event it should not be possible to 
create spacetime singularities. 

We have actually seen that a natural minimal 
length scale of the order of the Planck scale appears in this way. This result is 
of course not new, and it is actually already at the basis of the spacetime 
uncertainty relations of~\cite{DFR}, which in turn can be implemented by assuming 
that classical spacetime is replaced by a suitable noncommutative manifold. 
However, here we have derived such length scale by solving exactly part of the 
semiclassical Einstein equations. 

Thus, even if our analysis is bound to the 
spherically symmetric scenario, we have found a result which does not hold only 
in flat spacetime. Unfortunately, from this result alone, it is not possible to 
deduce the commutation relations of the quantum coordinates of events in a curved 
spacetime.

Nevertheless, in the last part of the paper, we have used that length scale in order to evaluate the influence of the quantum nature of spacetime on some expectation values of products of fields, by generalizing to a flat FRW spacetime the result obtained on Minkowski spacetime, where states of optimal localization are used
to define the product.
We have actually seen that, considering a simple cosmological model where the matter is described by a conformally coupled scalar field which mimics ordinary radiation, the scaling behavior of the energy density is significantly modified.   
Hence, 
taking into account the back reaction on the curvature of the modified stress tensor 
close to the initial singularity, a power law inflationary scenario arises.
Furthermore, in this simple model the form of the initial singularity changes in such a way that the usual horizon problem disappears.
We stress the fact that this result is not a consequence of a particular choice of the dynamics of the considered field, as in standard inflationary models. Rather, our field being simply a free one, it appears just as a consequence of the quantum nature of spacetime 
 which implies the existence of a new length scale, namely Planck length. Because of this, it can be expected that this feature is preserved also when the other approximation employed here, the use of a KMS state for the free field, is removed.

Up to now, this last observation is merely the result of 
extrapolations, which  employ the theory of 
quantum fields on classical curved spacetime together with the idea of optimal 
localization induced by the quantum nature of spacetime. In order to obtain 
further and more stringent results in this direction it seems necessary to 
address the problem of the construction of a Quantum Spacetime modeled on a 
general curved manifold and of a full-fledged quantum field theory on it.

\section*{Acknowledgments}
It is a pleasure to thank R. Brunetti and V. Moretti for 
useful discussions about the content of this paper and also L. Tomassini and S. Viaggiu for their comments on the first version.

The work of N.\ P.\ and G.\ M.\ has been supported in part by the ERC Advanced Grant 227458 ``Operator Algebras and Conformal Field Theory''.

\appendix

\section{Restriction of the vacuum state of a massless theory to a null cone in Minkowski spacetime}\label{Restriction}

In this appendix we would like to show that the continuity condition \eqref{continuity}, given in hypothesis 3 of Theorem 
\ref{singularity}, holds for the Minkowski massless vacuum $\om$. These computations are inspired by work in progress of R.~Brunetti and V.~Moretti~\cite{BM}, which is our pleasure to thank for discussing with us their results. We reproduce them here for the convenience of the reader.

Let us specify the following coordinate system on 
Minkowski spacetime 
$$
ds^2=-dv du +r^2 d\bS^2,
$$ 
where $v$ and $u$ are null coordinates and $r=(v-u)/2$. Furthermore, we have chosen the 
coordinates in such a way that the cone $\cC_0$ is the set of points that 
satisfy the conditions $v>0$ and $u=0$. We are thus interested in giving a 
continuity condition for $\om_2(f,g)$, i.e., the two-point function of the Minkowski vacuum of a massless 
scalar field evaluated on two smooth functions 
$f$ and $g$ supported in the future of the null cone $\cC_0$. We recall now that the form of its 
integral kernel is 
$$ 
\om_2(x,y) := \lim_{\epsilon \to 
0^+}\frac{1}{4\pi^2} \frac{1}{\sigma(x,y) + 2i\epsilon \at x_0-y_0 \ct }, 
$$ 
where 
$\sigma$ is the squared geodesic distance between $x$ and $y$ and $2 
x_0=u+v$. We would like to write $\om_2(f,g)$ as a functional which 
acts on functions defined on $\cC_0$. To this end, 
let us introduce $\Omega_\Si$, the standard symplectic form computed on 
some Cauchy surface $\Si$, and let $\psi_f=\Delta(f)$ and 
$\psi_g=\Delta(g)$, namely the solutions of the wave equation associated 
with $f$ and $g$. We can now introduce the operator 
$$ 
W_\Si(\psi_g)(x):= 
\Om_\Si(\om_2(x,\cdot),\psi_g), 
$$ 
which maps real smooth solutions of the 
wave equation in real smooth solutions. With this operator, we can write 
\beq\label{statobordo} \om_2(f,g)=\Om_\Si(\psi_f,W_\Si(\psi_g))\;. \eeq 
Notice that both $W_\Si$ and $\Om_\Si$ do not depend on the particular 
choice of the Cauchy surface $\Si$. 
Using Stokes Theorem, we can thus 
deform the hypersurface $\Si$ in such a way that it coincides with $\cC_0$ 
at least on $S=(J^-(\mathrm{supp}\, f) \cup J^-(\mathrm{supp}\, g))\cap \cC_0 $. 
Furthermore, because of the Huyghens principle the tip of the cone is not 
contained in $S$ and thus the support of the integrand in $W_\Si$ is 
bounded. %
Let $\psi_1$ and $\psi_2$ be two compactly supported smooth 
functions on $\cC_0$, we can then explicitly write their symplectic 
product as $$ \Om_{\cC_0}(\psi_1,\psi_2):= \frac{1}{4} \int_{\cC_0} \aq 
\frac{\pa\at v \psi_1\ct}{\pa v} v \psi_2-
 v \psi_1 \frac{\pa\at v \psi_2\ct}{\pa v} \cq dv\;d\bS^2= \frac{1}{2} 
\int_{\cC_0} \frac{\pa\at v \psi_1\ct}{\pa v} v \psi_2 \; dv\;d\bS^2, $$ 
where $d\bS^2=\sin \theta d\theta d\varphi$ is the standard measure on the 
unit sphere. With this in mind, we can from now on consider $W_{\cC_0}$ as 
a map from $C_0^\infty(\cC_0)$ to $\cD'(M)$ defined as 
\begin{gather*} W_{\cC_0}(\psi)(x):= 
\lim_{\epsilon\to 0^+ } \frac{1}{2\pi^2}\int_{\cC_0} \frac{\pa}{\pa{v'}} \at r' 
\frac{1}{\sigma_\epsilon(x,x')} \ct \psi(x' ) r' \sin(\theta') dv'd\theta' 
d\varphi' \;, \end{gather*} where the derivatives are taken in the weak 
sense and $x=(u,v,\theta,\varphi)$, $x'=(v',\theta',\varphi')$, $r'=v'/2$. 
In order to simplify that expression, 
let us start by recalling the explicit expression of 
$\sigma_\epsilon(x,x')$ which can be written as $-(u-i\epsilon ) 
(v-v'-i\epsilon ) + rr'(1-\cos\theta') $. Here 
we have chosen the 
spherical coordinates of $x'$ in such a way that when $\theta'=0$ the 
angle between $x$ and $x'$ vanishes. 
Inserting this in the previous equation we obtain \begin{gather*} 
W_{\cC_0}(\psi)(x)= \lim_{\epsilon\to 0^+ } \frac{1}{2\pi^2\; r}\int_{\cC_0} 
\frac{\pa}{\pa{v'}} \frac{\pa}{\pa \theta'} \log \big[ -(u-i\epsilon ) 
(v-v'-i\epsilon ) + rr'(1-\cos\theta')  \big]
  \psi(x' ) r' dv'd\theta' d\varphi'  \;.  
\end{gather*}
Now we shall integrate by parts in the $\theta'$ variable. 
Thus we end up with two boundary terms and an integral,
namely
\begin{gather*}
W_{\cC_0}(\psi)(x)=  - \lim_{\epsilon\to 0^+ }     
\frac{1}{2\pi^2\; r}\int_{\bR^+\times\bS^1}   \frac{\pa}{\pa{v'}}
\at  \log (u-i\epsilon ) +  \log (v'-v+i\epsilon )    \ct 
  \psi(v',0,\varphi' ) r' dv' d\varphi'  \; 
  \\
 + \lim_{\epsilon\to 0^+ }     
\frac{1}{2\pi^2\; r}\int_{\bR^+\times\bS^1}   \frac{\pa}{\pa{v'}}
\log  \big[  -(u-i\epsilon ) (v-v'-i\epsilon ) + 2 rr'   \big]
  \psi(v',\pi,\varphi' ) r' dv' d\varphi'  \; 
  \\ 
  - \lim_{\epsilon\to 0^+ }
  \frac{1}{2\pi^2\; r}\int_{\cC_0}  
  \frac{\pa \psi(x' )r'}{\pa \theta'} 
  \frac{\pa}{\pa{v'}}
 \log  \big[  -(u-i\epsilon ) (v-v'-i\epsilon ) + rr'(1-\cos\theta')   \big]
    dv'd\theta' d\varphi' \;.  \end{gather*} In the first integral, 
$\theta'=0$ correspond to the standard polar coordinate singularity, hence 
$\psi(v',0,\varphi')$ does not depend on $\varphi'$. We can thus perform 
the integration in $d\varphi'$. Moreover, after taking the $v'$-derivative, we 
change the angular coordinates for $x'$ by means of a rotation in order to 
have the same angular coordinates for $W_{\cC_0}(\psi)$ and for $\psi$ and 
we obtain \begin{gather}\label{restriction} 
W_{\cC_0}(\psi)(u,v,\theta,\varphi)= - \frac{1}{2\pi\;r} \lim_{\epsilon\to 
0^+ } \int_{\bR^+}
 \frac{\psi(v',\theta,\varphi ) v'}{(v'-v+i\epsilon )}    
    dv'  \; + I_B(u,v,\theta,\varphi)+ I(u,v,\theta,\varphi),  
\end{gather}
where $I_B$ and $I$ are the contributions due to the last two integrals.
We shall now use this expression to  evaluate $\om_2(f,g)$ as $\Om_{\cC_0}(\psi_f,W_{\cC_0}(\psi_g))$. 
Notice that the contributions to $\om_2(f,g)$ due to both $I$ and $I_B$ in 
$W_{\cC_0}(\psi_g)$ vanishes. We shall prove it for $I$ in some detail, the case involving $I_B$ can be dealt with in an analogous way.
To this end we need to consider the restriction of $W_{\cC_0}(\psi_g)$
 on $\cC_0$ and in particular
$$
I(0,v,\theta,\varphi):=  \lim_{u\to 0 } \lim_{\epsilon\to 0^+ }
  \frac{1}{2\pi^2\; v}\int_{\cC_0}  
  \frac{\pa^2 \psi_g(x' ) v'}{\pa v' \pa \theta'} 
 \log  \big[  -(u-i\epsilon ) (v-v'-i\epsilon ) + rr'(1-\cos\theta')   \big] 
    dv'd\theta' d\varphi'  \,
$$ 
where we have performed a $v'$-integration by parts whose boundary 
terms vanish because of the support properties of $\psi_g$. Since 
the $\log\big(rr' (1-\cos\theta')\big)$ is integrable, we can take the 
limit in the opposite order in the expression 
above. In this way, we obtain 
$$
I(0,v,\theta,\varphi)= 
 \int_{\cC_0}  
\frac{1}{2\pi^2\; v}  \frac{\pa^2 \psi_g(x' ) v'}{\pa v' \pa \theta'} 
  \aq \log v  + \log v' + 
 \log  (1-\cos\theta') \cq   
    dv'd\theta' d\varphi'  \;.  
$$
Notice that $v I(0,v,\theta,\varphi)$ is constant in $v$ and thus, since $\pa_v \at v I\ct=0$, it cannot contribute to $\Om_{\cC_0}(\psi_f,W_{\cC_0}(\psi_g))$. 
By a similar argument, the same conclusion can be drawn also for $v I_B(0,v,\theta,\varphi)$.
We end up with 
$$
\om_2(f,g)=   \frac{1}{2\pi} \lim_{\epsilon\to 0^+ }     
\int_{\bR^+\times\bR^+\times\bS^2}  
 \frac{ \psi_f(v,\theta,\varphi)v\;\psi_g(v',\theta,\varphi ) v'}{(v'-v+i\epsilon )^{2}}    
    dv dv' d\bS^2\;. $$ The last expression has been already studied in 
the literature, see for example \cite{DMP,DPP,Moretti06}.
It gives 
rise to a distribution which enjoys the following continuity condition: 
$$ 
|\om_2(f,g)| \leq \frac{1}{2\pi} \at \| v \psi_f \|_2\; \| \pa_v (v \psi_g) \|_2 + 
\|v \psi_g \|_2 \; \| \pa_v (v \psi_f)  \|_2 \ct,
$$ 
where the norms on the 
right hand side are the $L^2(\cC_0,dvd\bS^2)$ norms. This holds for every $f$ and $g$ with compact support 
contained in the future of $\cC_0$, and thus it reduces to equation 
\eqref{continuity} for the specific case.

\end{document}